\newif\ifSC
\SCtrue
\ifSC
\documentclass[onecolumn,draftclsnofoot,12pt]{IEEEtran}
\else
\documentclass[10pt,twocolumn]{IEEEtran}
\fi
\ifCLASSINFOpdf
\else
\usepackage[dvips]{graphicx}
\fi
\ifCLASSINFOpdf
\usepackage[pdftex]{graphicx}
\DeclareGraphicsExtensions{.pdf}
\else
\fi

\usepackage{mathtools}
\usepackage{graphicx}
\usepackage{amsthm}
\usepackage{cite}
\usepackage{color}
\usepackage{bbm}
\usepackage{amsmath,amssymb,amsfonts}
\usepackage{tcolorbox}

\newcommand{\degradationTime}{T_\mathrm{d}}
\newcommand{\ind}[1]{\mathbbm{1}\left(#1 \right)}
\newcommand{\rthree}{\mathbb{R}^3}

\newcommand{\pH}[1]{\mathrm{p_H}\left(#1\right)}
\newcommand{\fH}[1]{f_\mathrm{H}\left(#1\right)}

\newcommand{\ball}[2]{\mathcal{B}\left(#1, #2\right)}
\newcommand{\wt}{\mathbf{B}_t}
\newcommand{\wtmu}{\mathbf{W}_t}
\newcommand{\x}{\mathbf{x}}
\newcommand{\y}{\mathbf{y}}

\newcommand{\vw}{v(t)}
\renewcommand{\vw}{\left|\set{K}_t\right|}
\newcommand{\vdw}{v_d(t)}
\renewcommand{\vdw}{\left|\set{M}_t\right|}
\newcommand{\dd}{\mathrm{d}}
\newcommand{\avg}[1]{\mathbb{E}_{\wt}\left[ #1 \right]}
\newcommand{\avgmu}[1]{\mathbb{E}_{\wtmu}\left[ #1 \right]}
\newcommand{\EDP}{\mathrm{p}_\mathrm{D}(t)}

\newcommand{\expects}[2]{\mathbb{E}_{#1}\left[ #2 \right]}

\newcommand{\pro}[1]{\mathbb{P}\left[ #1 \right]}
\newcommand{\prob}[1]{\mathbb{P}\left[ #1 \right]}
\newcommand{\vt}{\avg{\vw}}
\newcommand{\vdt}{\avg{\vdw}}
\renewcommand{\vdt}{\avgmu{\vdw}}
\newcommand{\ie}{{\em i.e., }}
\newcommand{\set}[1]{\mathsf{#1}}
\renewcommand{\Xi}{\set{K}_t}
\newcommand{\tc}{t_\mathrm{c}}
\newcommand{\Ximu}{\set{M}_{t}}
\newcommand{\dist}[1]{\|#1\|}
\newcommand{\erf}{\mathrm{erf}}
\newcommand{\erfc}{\mathrm{erfc}}

\newcommand{\subsectioninline}[1]{\noindent\textbf{#1}:\ }

\newcommand{\difs}{\mathrm{ds}}
\newcommand{\norm}[1]{\left\lVert#1\right\rVert}
\newtheorem{theorem}{Theorem}
\newtheorem{lemma}{Lemma}
\newtheorem{coro}{Corollary}
\newtheorem{remark}{Remark}

\begin{document}

\title{\huge Detection Probability in a Molecular Communication via Diffusion System with Multiple Fully-absorbing Receivers}

\author{Nithin V. Sabu and Abhishek K. Gupta
\thanks{The authors are with the department of Electrical Engineering, Indian Institute of Technology Kanpur, Kanpur, India 208016. (Email: nithinvs@iitk.ac.in and gkrabhi@iitk.ac.in). This research was supported by the  Science and Engineering Research Board   (India) under the grant SRG/2019/001459 and IITK  under the grant IITK/2017/157.}
\vspace{-0.1in}}

\maketitle

\begin{abstract}
In this letter, we consider a 3D molecular communication via diffusion system (MCvDS) with a single point transmitter and multiple fully-absorbing spherical receivers whose centers are distributed as a Poisson point process (PPP) in the medium. We derive the probability that a transmitted molecule hits any of the receivers within time $ t $. We consider both degradable and non-degradable molecules. We verify the analysis using particle-based simulation. The framework can be used for various applications, e.g., to derive event detection probability for  systems where the IMs are transmitted to convey the occurrence of a particular event to trigger reactions at receivers or can be used as channel models for such systems.
\end{abstract}


\section{Introduction}
In an MCvDS, molecules  carrying information from transmitter bio-nanomachine (TBN) to receiver  bio-nanomachine (RBN), propagate  in the medium via diffusion \cite{Nakano2011}. These molecules are termed as \textit{information molecules} (IMs). The information that needs to be conveyed can be a bit stream, or intimation of a particular event occurring at  TBN, or a control to trigger a reaction at RBN. The RBN usually consists of receptors that bind with IMs to detect the transmission and decode the information.   
 
 MCvDSs with single/multiple TBNs and single RBN have been studied in the past literature. The channel for an MCvDS with a point transmitter and a spherical fully-absorbing (FA) receiver was derived in \cite{Yilmaz2014}. A fully-absorbing receiver is the one that absorbs all the IMs hitting its surface and decodes based on the count of absorbed molecules. The channel  for a similar system, but with degradable IMs,  was discussed in \cite{Heren2015}.
Due to the limited capabilities of bio-nanomachines (BNs), multiple RBNs should co-operate to perform complex tasks or to improve the reception. This poses the requirement of channel models involving multiple FA-RBNs. A 1D two-receiver system was studied in \cite{Guo2016}, where the authors have derived the  expression for the fraction of molecules absorbed by the two FA-receivers. The fraction of molecules absorbed by each of the absorbing receiver located in a 1D two-receiver system was derived in \cite{Huang2019}. The fraction of molecules absorbed by each receiver in a 3D two-receiver system was derived in \cite{Kwack2018}. The hitting probability of an IM at each FA receiver in a 3D two-receiver system was studied via simulations in \cite{Lu2016}. For a system with multiple absorbing receivers, most works in the literature used empirical formula for channel response that is obtained using data fitting methods \cite{Bao2019,Koo2016,Damrath2017}. 
 A 3D MCvDS with multiple FA-RBNs with and without degradable IMs has not been studied analytically in the past, which is the focus of this letter.

In this work, we develop an analytical framework for an MCvDS with a single TBN and multiple FA-RBNs with and without degradable IMs, using tools from stochastic geometry \cite{Andrews2016a}. Each RBN is spherical, with its center located uniformly in the medium. Hence, RBNs can be modeled as a Boolean Poisson process. We present the probability that an IM emitted by a point TBN hits \textit{any one} of the RBNs within time $t$. Characterizing the hitting probability for each RBN is difficult and out of the scope of this paper.
 The presented hitting probability gives the probability that at least one  RBN detects the transmission and trigger a specific reaction. Based on the hitting probability, we derive event detection probability which is the probability that at least one out of $ N $ transmitted IMs are absorbed by any of the RBNs. We consider an example to show the applicability of the proposed framework for systems that perform particular tasks when the IMs absorbed jointly by multiple RBNs cross a threshold value.
The proposed model can serve as a channel model for an MCvDS with single TBN and multiple RBNs. The proposed model can also be applied to a system where the TBN transmits IMs to convey the information that a particular event has occurred to trigger some reaction at RBNs. The model is suitable for applications where the RBNs are deployed to detect a single event such as detection of cancerous cell based on the bio-markers emitted by it \cite{felicetti}, or detection of toxic gas in the environment. After detection, RBN can trigger some defensive actions or  communicate this information to a central node depending on the implementation. Event detection probability derived in this letter can be used in the  design of such a system, e.g., to compute optimal deployment of RBNs that maximizes the detection. The analysis can also be extended to study systems with ISI.\\
\textbf{Notation}: $ \ball{\x}{a} $ represents a ball of radius $a$ centered at the location $\x$. $\set{A}\oplus\set{B}$ represents the Minkowski sum of the two sets $\set{A}$ and $\set{B}$. $|\set{A}|$ is the volume of $\set{A}$. $\phi$ denotes null set.
\begin{figure}
	\centering
	\includegraphics[width=0.7\linewidth]{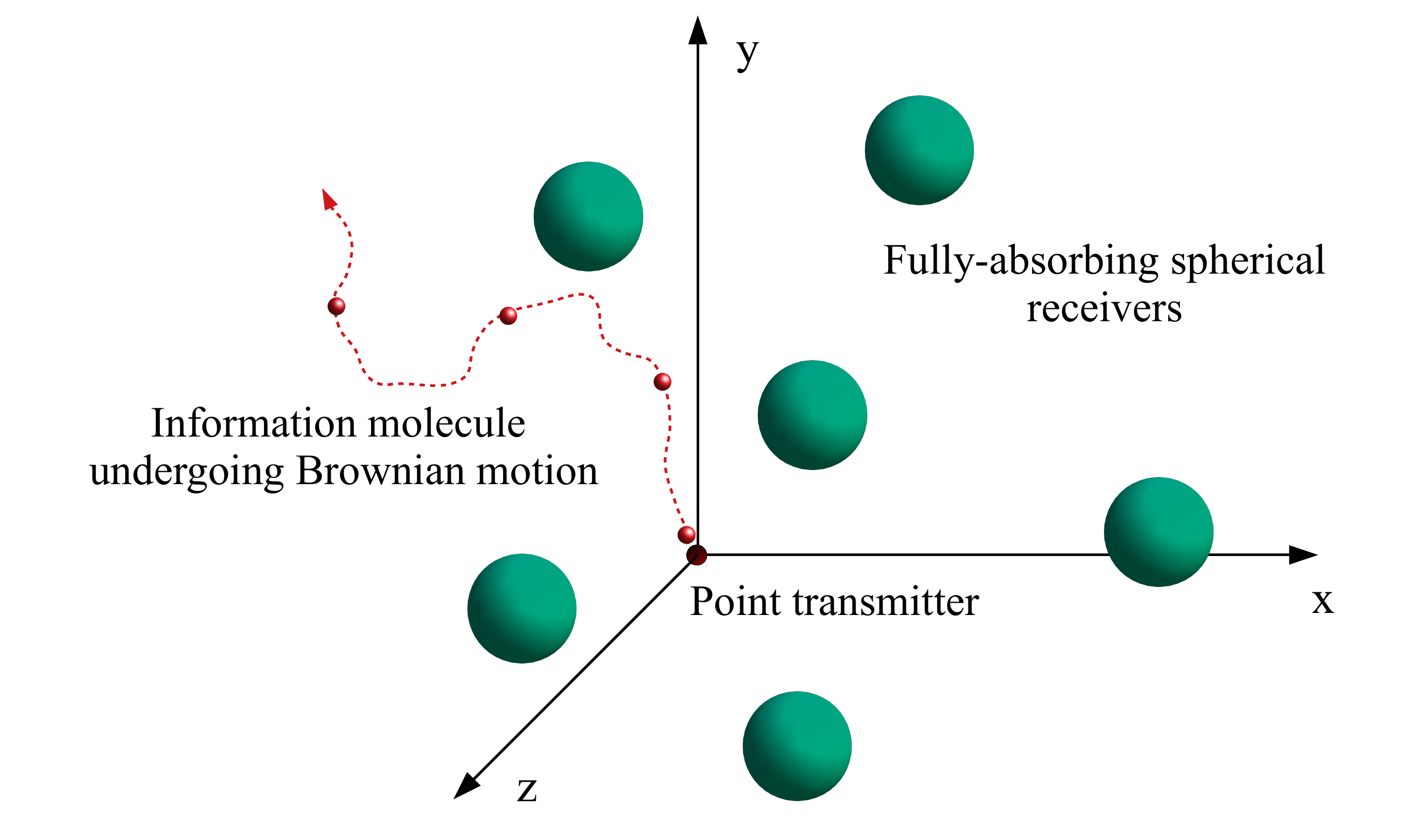}
	\caption{A MCvDS with a point TBN at the origin and multiple RBNs.}
	\label{fig:sm}\vspace{-.15in}
\end{figure}
\section{System Model}

In this letter, we consider an MCvDS in a 3D medium without flow, as shown in Fig. \ref{fig:sm}. The system consists of a point TBN and multiple RBNs, which are fully-absorbing spherical receivers of radius $a$. {\color{black}We can also easily extend the system model with RBNs of different radii by considering superposition of PPPs each having RBNs of different radius.} TBN emits IMs to the medium, which propagates through the medium via Brownian motion. The IMs are detected at an RBN when they hit the surface of a particular receiver.

\subsectioninline{Network Model}
Without loss of generality, we assume that  the TBN is located at the origin. RBNs are modeled as a  Boolean Poisson process $\Psi$, where the centers of RBNs are distributed as a uniform Poisson point process (PPP) $\Phi=\{\x_n:n\in\mathbf{N}\}$ in the region $\set{S}=\rthree\setminus\ball{0}{a}$ with constant density $\lambda>0$. Each RBN is a fully absorbing spherical receiver with $i$th RBN modeled as $\ball{\x_i}{a}$. Therefore,
\begin{align*}
\Psi=\bigcup_{\x_n\in\Phi}\x_n+\ball{0}{a}.
\end{align*}
Here, the region $\set{S}$  ensures that the TBN does not lie inside any RBN.  
We assume that, at time $ t=0 $ the TBN emits $N$ IMs to the propagation medium to communicate the occurrence of an event to the RBNs. These IMs propagate in the medium. After reaching near an RBN, the receptors present in the surface of RBN bind with the IMs. The transmit bit can be estimated based on the detection of at least one IM at any of the RBN.

\subsectioninline{Propagation Model}
Each IM propagates in the medium via Brownian motion independent of the other IMs. The diffusion coefficient of the IM in the propagation medium is $D$.
Let $\wt=\{\y_s,\ 0\leq s \leq t\}$ be the locus of points visited by an IM during the interval $[0,t]$. Here $\y_s$ is the location of the IM at time $s$. $\y_s$  is known as the Weiner process, which is a random process. $\wt$ is termed as the Brownian path.

\subsectioninline{Molecular Degradation}
Due to the presence of other molecules (either naturally present or added intentionally) in the medium, the IM degrades over time as a result of its interaction with them. We consider the first-order degradation with degradation rate constant $\mu$. This means that the probability that an IM does not degrade until time $t$ is $e^{-\mu t}$.
Let $\degradationTime$ is the degradation time of the IM. 
Hence, the probability that  the degradation occurs after time $t$ is
\begin{align}
\prob{\degradationTime>t}=\exp(-\mu t).\label{e14}
\end{align}The path $\wtmu$ of a degradable IM is   defined as
\begin{align*}
\wtmu=\begin{cases}
\wt & \text{if } t<\degradationTime\\
\mathbf{B}_{\degradationTime}& \text{if } t>\degradationTime
\end{cases}.
\end{align*}  
\section{Hitting Probability of an IM}
In this section, we consider the motion of only one IM and study the probability that a particular IM hits the surface of any of the RBNs within time $t$. We term this probability as the hitting probability $\pH{t}$ of this IM. Let $T$ denote the hitting time \ie the time at which the IM hits the surface of any RBN. 
$$\pH{t}=\prob{T\le t}.$$
 \subsection{Hitting probability of IM when there is no degradation}
 \begin{figure}
 	\centering
 	\includegraphics[width=0.7\linewidth,trim=0 0 0 0,clip]{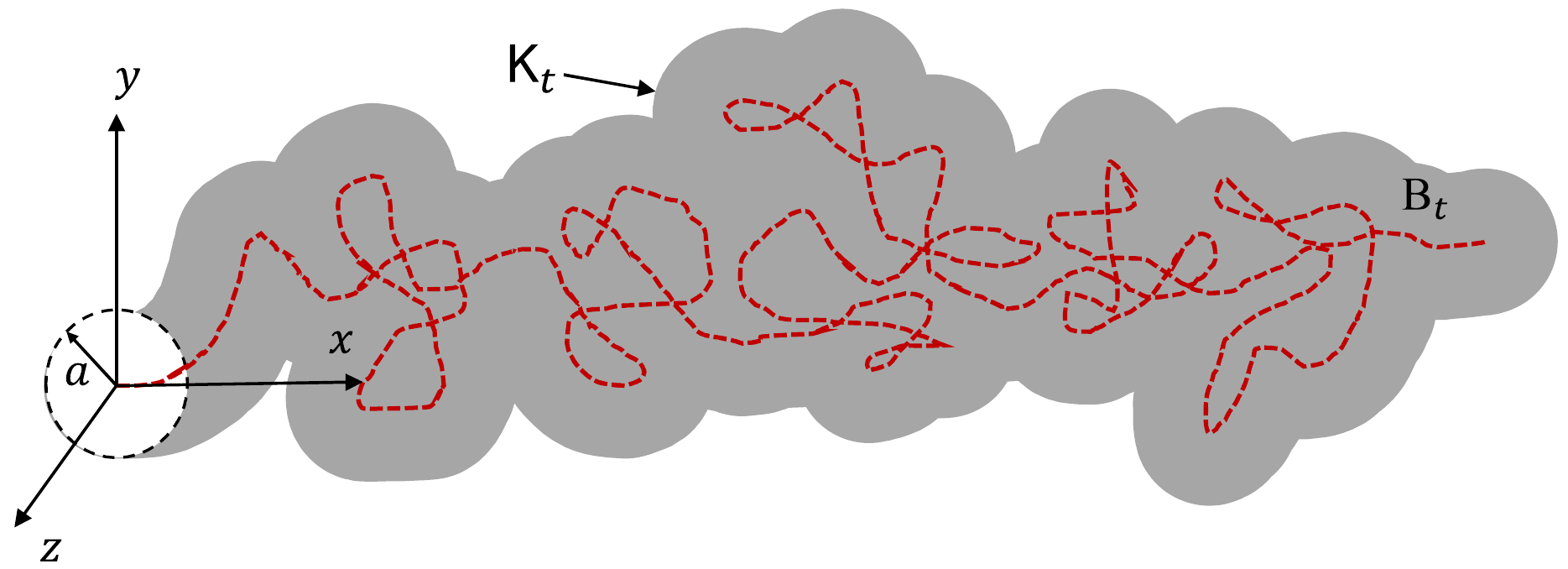}
 	\caption{The shaded part represents the region $\Xi$. $\wt$ represents a realization of the Brownian path traced by an IM, which is emitted by the TBN. \vspace{-.2in}}
 	\label{fig:volume}
 \end{figure}
 We first consider the case when IM is not degradable \ie $\mu=0$.  Recall that $\wt$ denote the Brownian path of the IM. The event $\set{E}_t$ that the IM hits the surface of a RBN located at $\x_i$ (\ie $\ball{\x_i}{a}$) within time $t$ is equivalent to the event 
\begin{align*}
\set{E}_t&=\{\wt \cap \ball{\x_i}{a}\ne \phi\}=\cup_{s\le t}\{\y_s\cap \ball{\x_i}{a}\ne\phi\}\\
&=\cup_{s\le t}\{\dist{\y_s-\x_i}\leq a\}=\cup_{s\le t}\{\x_i\cap \ball{\y_s}{a}\ne\phi\}\\
&=\left\{\x_i \cap \left(\cup_{s\le t}\ball{\y_s}{a}\right)\ne\phi\right\}\\
&=\left\{\x_i\cap \left(\wt\oplus \ball{0}{a}\right)\ne\phi\right\}\\
&=\left\{\x_i\in \left(\wt\oplus \ball{0}{a}\right)\right\}.
\end{align*}
The event of IM hitting on the surface of any RBN in $\Psi$ is equivalent to the event that at least one point $\x_i$ in $\Phi$  lies inside the region $\left(\wt\oplus\ball{0}{a}\right)$. Since all $\x_i$s lie outside $\ball{0}{a}$, it is equivalent to the event that at least one point $\x_i$ in $\Phi$  lies inside the region $\Xi=\left(\wt\oplus\ball{0}{a}\right)\setminus \ball{0}{a}$.  The shaded region in Fig. \ref{fig:volume} represents $\Xi$. Hence, hitting probability of the IM  is equal to
 \begin{align}
\pH{t}&=\prob{\set{E}_t}=1-\avg{\prob{\x_n\not\in \Xi,\ \forall\ \x_n\in\Phi}}\nonumber\\
&=1-\avg{\exp\left(-\lambda\vw\right)},\label{e3}
\end{align}
where the last step is due to void probability of PPP \cite{Andrews2016a}.  $\vw$ is a function of the Brownian path $\wt$ and hence is a random variable. Note that \eqref{e3} gives the exact hitting probability for  FA RBNs and is the same as the death probability of a Brownian particle in the presence of multiple traps derived in \cite{Berezhkovskii1991a}. As \cite{Berezhkovskii1991a}, applying cumulant expansion \cite{kenney1951cumulants} in \eqref{e3} gives
 \begin{align}
\pH{t}=1-\exp\left(\sum\nolimits_{n=1}^{\infty}\kappa_n{\left(-\lambda\right)^n}/{n!}\right),\label{e4}
\end{align}
 where $\kappa_n$ is the $n^{\text{th}}$ cumulant of $\vw$. Recall that, for a random variable, its first cumulant is its mean. The first cumulant \ie the mean of $\vw$ is given in the following Lemma.
 \begin{lemma}\label{lemma1}
 The mean volume of $\vw$ is 
 $$\kappa_1=\vt=4\pi Dat+8a^2\sqrt{\pi D t} .\label{eq:kt}$$
 \end{lemma}
 \begin{proof}
 See Appendix \ref{app:A} for the proof.
 \end{proof}

Further, the second cumulant of a random variable is its variance, and the third cumulant is its third central mean, which may be challenging to obtain. However, we can approximate $\pH{t}$ in \eqref{e4} by retaining only the first cumulant and ignoring higher-order cumulants. \ie
 \begin{align}
 \pH{t}\approx 1-\exp\left(-\lambda \vt\right)\label{e5}.
 \end{align} 
Note that this is also an upper bound due to Jensen's inequality \cite{Berezhkovskii1991a}. The approximation used to obtain \eqref{e5} is valid when $\lambda \times \frac43\pi a^3 \ll 1$ and $t$ is not very large.
 Now, substituting value of $\vt$ from Lemma \ref{eq:kt} in \eqref{e5}, we get the following Theorem 1 \cite{Berezhkovskii1991a}.
  \begin{theorem}\label{thm1}
  	For a non-degradable IM, the hitting probability of  an IM on any of the RBN in $\Psi$ within time $t$ is given as
  	\begin{align}
  		\pH{t}= 1-\exp\left(-4\pi\lambda Dat-8a^2\lambda\sqrt{\pi D t}\right)\label{e11}.
  	\end{align}
  \end{theorem}
From \eqref{e11}, we can see that the hitting probability is a function of density and radius of RBN, and the diffusion coefficient $D$. We now compute the time constant of the hitting process. The time constant $ t_c $ of a process is defined as the time  it takes for the process to be completed by a factor $k=1-1/e$. It represents the time-scale of the process \ie the order of magnitude of time in which the process occurs.  From its expression, we can understand which parameters affect the rate of the process and under which conditions.
\begin{remark}
	The time-constant for the hitting process (\ie the time it takes for the hitting probability to attain $1-1/e$ of its maximum value) is 
	\begin{align}
	t_\mathrm{c}=\frac{a^2}{\pi D}\left(1-\sqrt{1+\frac{1}{4a^3\lambda}}\right)^2,
	\label{rem1a}
	\end{align}	
    and is in the order of 
	\begin{align}
		t_\mathrm{c}=\mathcal{O}\left(\frac{1}{4\pi\lambda Da}\min\left(1,\frac{1}{16a^3\lambda}\right)\right).
		\label{rem1}
	\end{align}
\end{remark}
\begin{proof}
	See Appendix \ref{app:B}.
\end{proof}

Let $\fH{t}$ be defined as the hitting rate of IMs at the RBNs. In other words, $\fH{t}\Delta t$ denotes the probability that IM hits any RBN in time interval $[t,t+\Delta t]$. It is equal to the derivative of hitting probability of an IM within time $t$ in \eqref{e11}  as 
\begin{align}
\fH{t}&=\left(4\pi\lambda Da+{4a^2\lambda\sqrt{\pi D/t}}\right)\times\exp\left(-4\pi\lambda Dat-8a^2\lambda\sqrt{\pi D t}\right)\ind{t\geq 0}\label{e13}.
\end{align}
\subsection{Hitting probability of degradable IM}
We now consider degradable IM with the degradation process, as described in Section II. We assume that the degradation of IM is independent of the motion of the IM.  
The event of IM hitting on the surface of any RBN in $\Psi$ before its degradation is equivalent to the event that at least one point $\x_i$ in $\Phi$  lies inside the region $\Ximu=\left(\wtmu\oplus\ball{0}{a}\right)\setminus \ball{0}{a}$.  
 \begin{lemma}\label{lemma2}
	For a degradable IM, the average volume of $\Ximu$ is given as $\vdt=$
	\begin{align}
	{4\pi a}{(D/\mu)}\left(1-\exp(-\mu t)\right)+4\pi a^2\sqrt{({D}/{\mu})}\ \erf\left(\sqrt{\mu t}\right)
	\end{align}
\end{lemma}
\begin{proof}
Note that $\Ximu=\Xi\ind{t<\degradationTime}+\set{K}_{\degradationTime}\ind{t>\degradationTime}$. Using Lemma \ref{lemma1} and the PDF of $ \degradationTime$, we get the desired result.
\end{proof}
Note that, $\vdt\leq \vt$. 
Hence, degradation limits the  volume growth of $\Ximu$ with time.
The hitting probability  of the IM on any RBN within time $t$ before its degradation is 
\begin{align}
\pH{\mu,t}=\int_{0}^{t}\fH{s} \prob{\degradationTime>s}\difs \label{e16}.
\end{align}
Substituting the value of $ \fH{t} $ from \eqref{e13} and $\prob{\degradationTime>t}$ from \eqref{e14} in \eqref{e16} gives Theorem 2.
\begin{theorem}\label{thm2}
For an IM with degradation, the hitting probability of an IM at any one of the RBNs within time $t$ is
	\begin{align}
	\pH{\mu,t}&=\left(1-\exp\left(-2\beta\sqrt{\alpha t}-\alpha t\right)\right)\times\left(1-\frac{\mu}{\alpha}\right)+\frac{\sqrt{\pi}\beta \mu}{\alpha}e^{\beta^2}\left(\erf\left(\beta+\sqrt{\alpha t}\right)-\erf\left(\beta\right)\right),\label{eq:e17}
	\end{align}
	where $\erf(.)$ is the error function, $\alpha=4a\pi \lambda D+\mu$ and $\beta=4a^2\lambda \sqrt{\frac{\pi D}{\alpha}}$.
\end{theorem}
\begin{coro}
	When $t\rightarrow \infty$, \eqref{eq:e17} denotes the probability of the IM reaching any one of the RBN before getting degraded, which is given as
	\begin{align}
	\pH{\mu,\infty}= 1-({\mu}/{\alpha})+{\sqrt{\pi}\beta }{(\mu/\alpha)}e^{\beta^2}\erfc(\beta),\label{e15}
	\end{align}
	where $\alpha$ and $\beta$ are same as in Theorem \ref{thm2}.
\end{coro}
\begin{remark}
	Substituting $\mu\rightarrow 0$ in \eqref{eq:e17} (\ie the IM is non-degradable), it can be shown that $\pH{0,t}=\pH{t}$.
\end{remark}
\begin{remark}
The time-constant for the hitting process with degradable IM is 
\begin{align}
t_\mathrm{c}\approx\frac{\beta^2}{\alpha}\left(1-\sqrt{1+\frac{1}{\beta^2}}\right)^2,
\label{rem3a}
\end{align}	
which is in the order of 
	\begin{align}
		t_\mathrm{c}&=\mathcal{O}\left(\frac{1}{4\pi\lambda Da}\min\left(\frac{4\pi\lambda Da}{4\pi\lambda Da+\mu},\frac{1}{16a^3\lambda}\right)\right).\label{rem3}
	\end{align}
\end{remark}
\begin{proof}
	See Appendix \ref{app:C}.
\end{proof}
\begin{figure}
	\centering
	\includegraphics[width=0.7\linewidth]{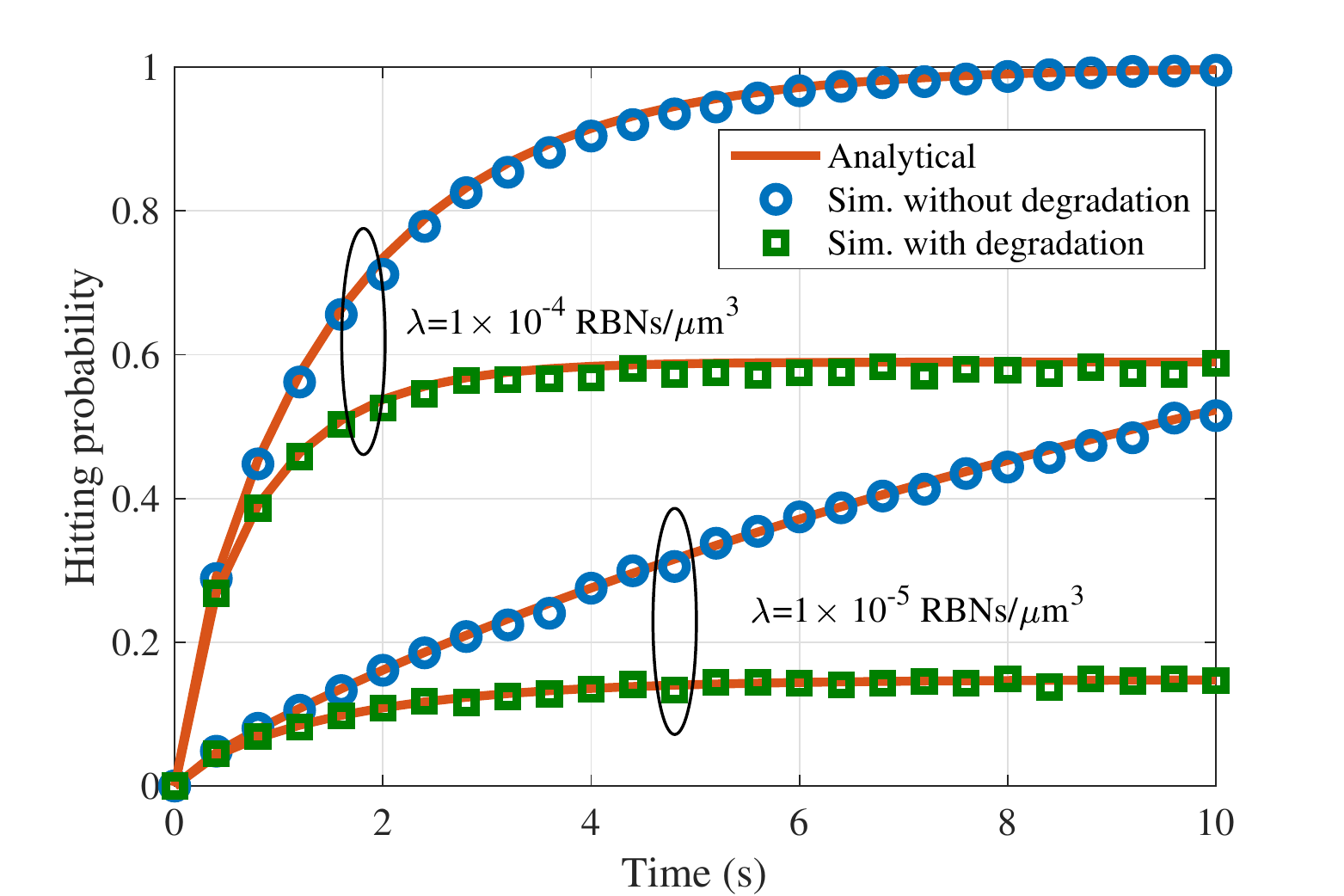}
	\caption{The hitting probability versus time for different values of RBN density and degradation  constant. Here, $a=5\mu$m and $D=100\mu$m$^2$/s.\vspace{-.3in}}
	\label{fig:f1}
\end{figure}
\begin{figure}
	\centering
	\includegraphics[width=0.9\linewidth]{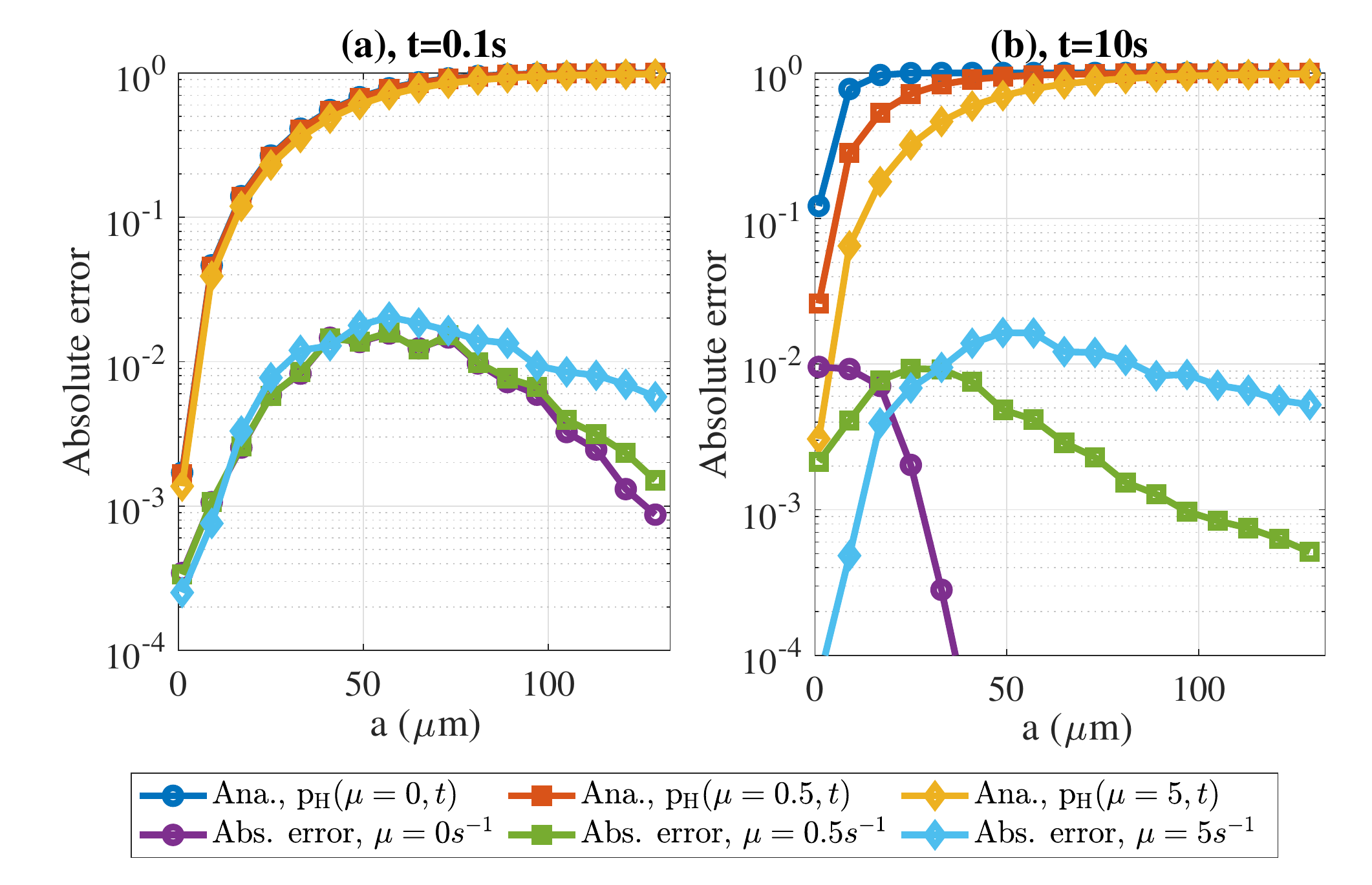}
	\caption{\color{black}Absolute approximation error versus $a$ for different values of $ t $ and $ \mu $. Here, step size $\Delta t=10^{-4}s,\  D=100\mu \text{m}^2/s,\ \lambda=1\times 10^{-5}\text{RBNs/}\mu \text{m}^3$.\vspace{-.2in}}
	\label{acc1}
\end{figure}
\begin{figure}
	\centering	
	\includegraphics[width=0.7\linewidth]{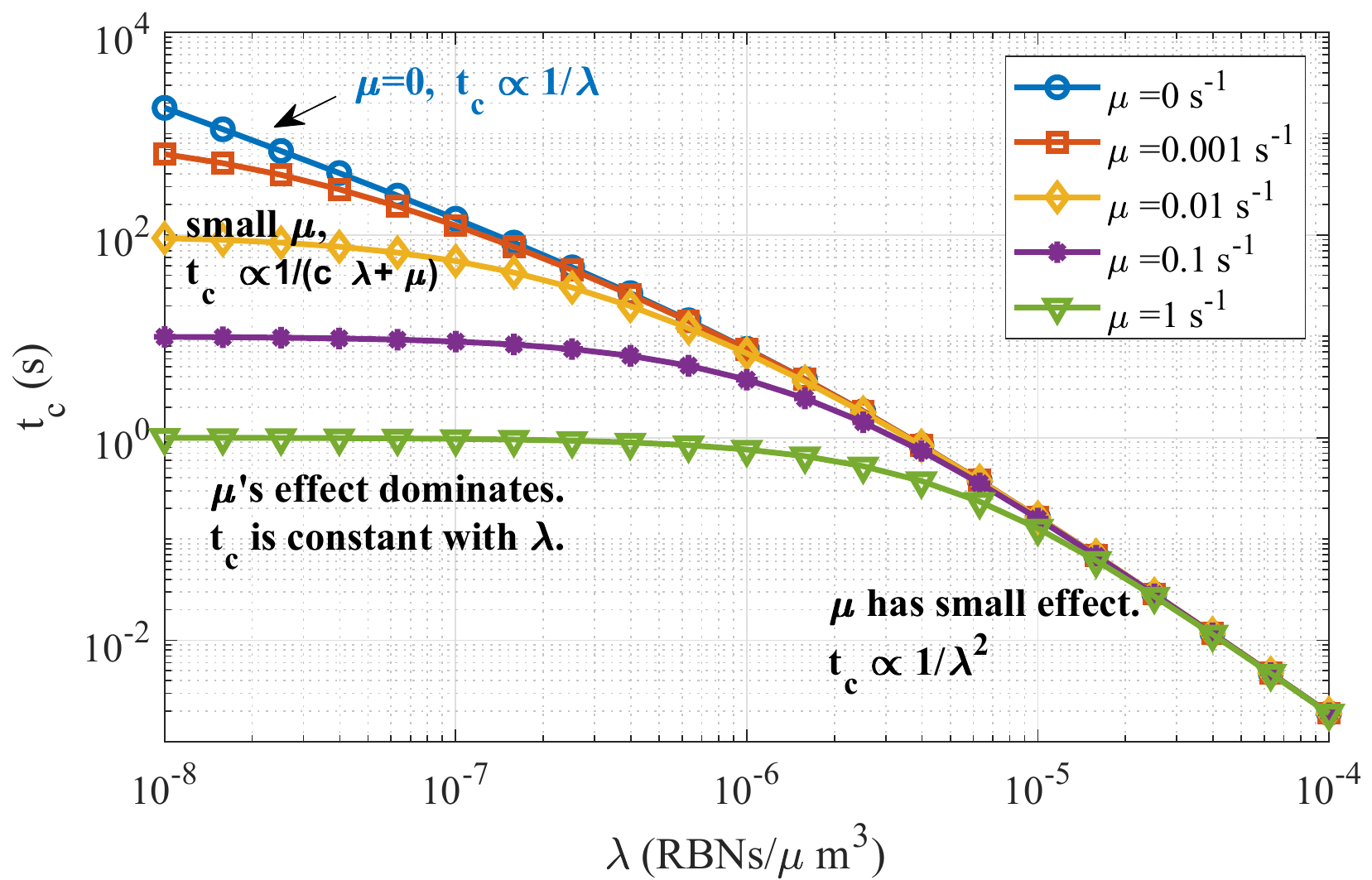}
	\caption{Variation of time constant $\tc$ with $\lambda$ for different values of $\mu$ showing the impact of parameters in various regions. Here, $D=100\mu \text{m}^2/s$, $a=40\mu \text{m}$.}
	\label{tc}
	\vspace{-.15in}
\end{figure}

Fig. \ref{fig:f1} validates the analytical results for hitting probability given in \eqref{e11} and \eqref{eq:e17}, using particle-based simulation. The simulation time step $\Delta t$ for particle-based simulation is $10^{-4}$s. 
The simulation is performed over $ 10^4 $ iterations. In each iteration, one IM is generated, and its movement is recorded. Note that, this is equivalent to a single step simulation with $ 10^4 $ IMs. 
For non-degradable IM, the hitting probability increases with time to eventually reach 1.  
For an MCvDS with degradable IM, the hitting probability is less than 1 even after a long time. This is because  IMs may degrade before hitting any RBNs. Although \eqref{e11} and \eqref{eq:e17} are approximate results, they match closely with the exact value.

Fig. \ref{acc1} shows the variation of the absolute approximation error (\ie $ | $analytical result$-$simulation result$ | $) with $a$ for different values of $ t $ and $ \mu $. It can be seen that  the error is very small, which shows the approximation is very accurate for a wide range of parameters.
	For improving the accuracy further, higher-order cumulants can be considered in \eqref{e11} and \eqref{eq:e17}.
	
		 Fig. \ref{tc} shows the variation of time-constant with respect to $ \lambda $ for different values of $ \mu $ using \eqref{rem1a} and \eqref{rem3a}. 
		 When $\mu=0$, at low values of $\lambda$,  $ \tc $ varies as $ 1/\lambda $ which is the first term in \eqref{rem1}. When $ \lambda $ is high, the second term of \eqref{rem1} is dominant and  $\tc$ decreases as $ 1/\lambda^2 $. This shows that at higher values, $\lambda$ has a larger impact on the hitting process.  
 For non-zero $\mu$, $\tc$ varies as $ 1/(c\lambda+\mu),\ c=4\pi D a  $ for low values of $\lambda$ according to \eqref{rem3}. This implies that at moderate to high values of $\mu$, $\lambda$ has less or no effect on $\tc$ and is mainly determined by the value of $\mu$. We can see that at $\mu=1s^{-1}$,  $ \tc $ is constant with $\lambda$ in this region.  However, when RBN density $ \lambda $ is high, the second term of \eqref{rem3} dominates. In this region,  $\tc$  varies as $ 1/\lambda^2 $, which means that $\lambda$ has a larger impact on the hitting process. However, $\mu$ has a very small impact on the hitting process.
\section{Event Detection Probability}
In this section, we derive the event detection probability $\EDP$, which is the probability that at least one IM out of $N$ IMs hits any of the RBNs within time $t$. Let $ z(t) $ be the number of IMs hitting any of the RBN. 
If we treat the event of hitting of an IM on any of the RBNs up to time $t$ as the success event (with probability $\pH{\mu,t}$),
then, $z(t)\sim \mathsf{Binom} \left(N,\pH{\mu,t}\right)$,  assuming the hitting probability of each IM is independent to each other.
This means that the probability that RBNs absorb more than $\eta$ IMs within time $t$ is 
\begin{align*}
	\pro{z(t)\geq\eta}=1-\sum_{k=0}^{\eta-1}\binom{N}{k}\pH{\mu,t}^k(1-\pH{\mu,t})^{N-k}.
\end{align*}
Hence, the event detection probability is given as $$\EDP=\prob{z(t)\geq 1}=1-{\left(1-\pH{\mu,t}\right)}^N.$$
Therefore, to ensure that the event is detected with probability $\EDP$, the transmitter should emit $N= \log(1-\EDP)/\log(1-\pH{\mu,t})$ number of IMs.
{\color{black} For applications requiring a central node, it has to be transparent (passive) to IMs to maintain the model's validity. Otherwise, the number of molecules received will be less than $ z(t)$. Combining this effect with the possibility of further loss in communication between the RBNs and the central node, {\color{black} $\EDP$ is an upper bound  on the actual detection probability in such applications.}} 
\begin{figure}
	\centering
	\includegraphics[width=0.7\linewidth]{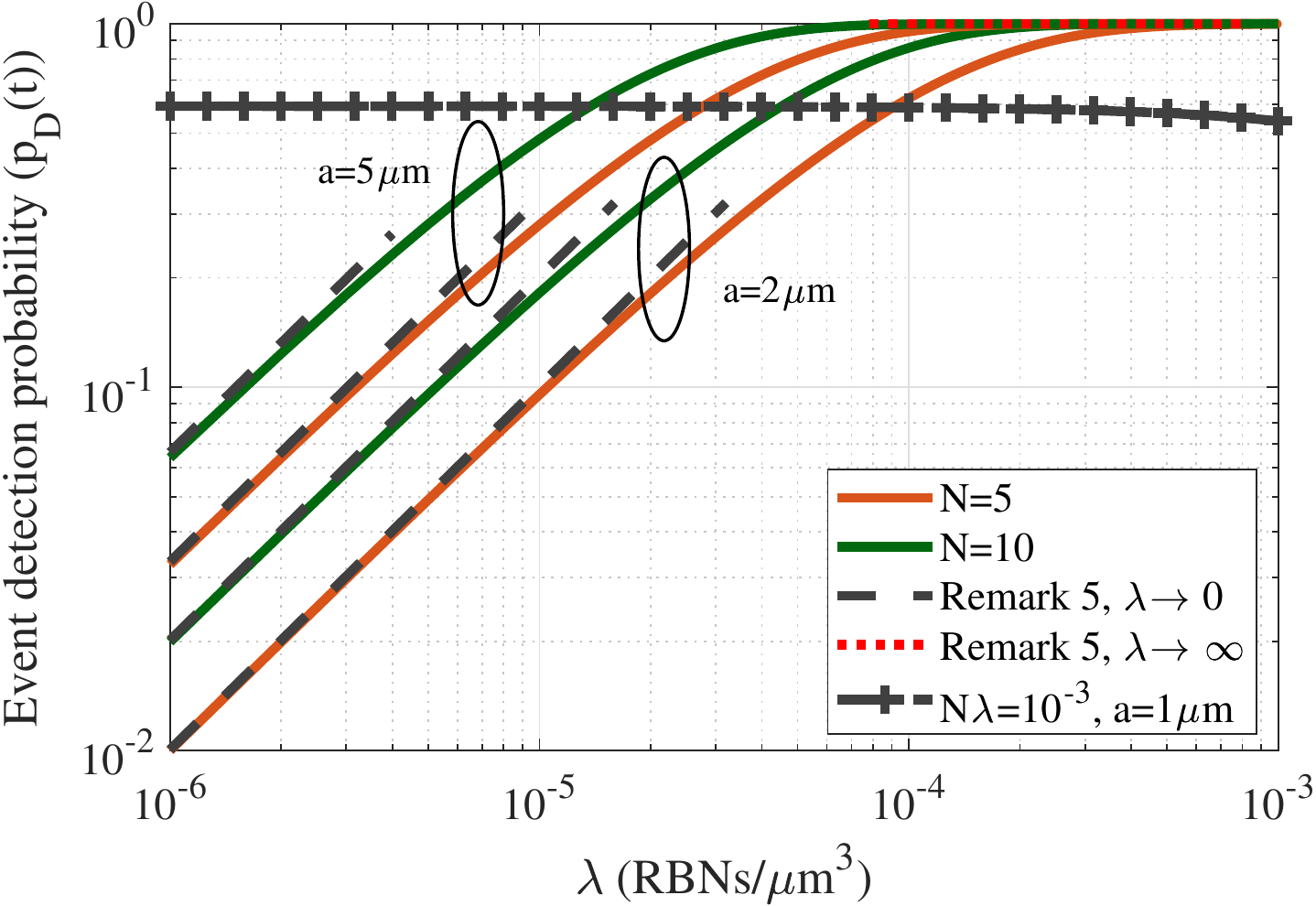}
	\caption{Event detection probability $\EDP$ versus RBN density $\lambda$. Here, $\mu=1$s$ ^{-1} $, $t=1$s and $D=100\mu$m$^2$/s.\vspace{-.5in}}
	\label{fig:f4n}
\end{figure}
\begin{remark}\label{rem:zeroDen}
If IMs are non-degradable, $\EDP=1-\exp\left(-4\pi N\lambda Dat-8a^2N\lambda\sqrt{\pi D t}\right)$. In other words, scaling either the RBN density $\lambda$ or the number of emitted molecules $N$ by the same factor has the same effect on the event detection probability. 
For example, doubling $\lambda$ or doubling $N$ has the same effect on $\EDP$.
\end{remark}

\begin{remark}\label{rem:infDen}
If IMs are degradable, a similar effect can be found under certain circumstances. When $\lambda\to 0$, the asymptotic  $\EDP$ 
 is linearly dependent on $N$ and $\lambda$ i.e. $\EDP=$ 
\begin{align*}
4a\pi N\lambda\left[(D/\mu){\left(1-e^{-\mu t}\right)}+a\sqrt{({D}/{\mu})}\ \erf\left(\sqrt{\mu t}\right)\right]+\mathcal{O}\!\left(\lambda^2\right)\!,
\end{align*}
which shows that scaling $\lambda$ or $N$ has the same effect on $\EDP$. 
However,  for large $\lambda$, scaling $\lambda$ or $N$ does not have the same effect. In particular, as $\lambda\rightarrow \infty$, 
	\begin{align*}
		\EDP=1-\left({0.0095\mu}/{(aD)}\right)^N\lambda^{-N}+\mathcal{O}\left({\lambda^{-2N}}\right).
	\end{align*} 
Therefore, if we fix $N\lambda=c$, $\EDP=-N\ln\left(N\right)-N\ln\left({0.0095\mu}/{(aDc)}\right)+\mathcal{O}\left(N^2\right)$.
	\end{remark}
Fig. \ref{fig:f4n} shows the variation of $\EDP$ with $\lambda$.
The limiting behavior for $\lambda\rightarrow 0$ and  $\lambda\rightarrow \infty$ is also shown as discussed in Remark \ref{rem:infDen}. 
 Fig. \ref{fig:f4n} also shows the variation of $\EDP$ with $\lambda$, when $ N\lambda=c $, where $ c $ is a constant. Fig. \ref{fig:f4n} verifies that,  $\EDP$ remains constant when $ \lambda $ is small and $\EDP$ falls (small dip at the end of curve) when $ \lambda $ is high as $N$ is very low. 
\appendices
\section{Derivation of Lemma 1}\label{app:A}
Note that $\vw$ is a function of $\wt$. Its average volume is 
\begin{small}
\begin{align}
\vt&=\expects{\wt}{\int_{\rthree} \ind{\y\in\Xi}\dd\y}\nonumber\\
&=\expects{\wt}{\int_{\rthree\setminus \ball{0}{a}} \ind{\y\in\wt\oplus\ball{0}{a}}\dd\y}\nonumber\\
&=\expects{\wt}{\int_{\rthree\setminus \ball{0}{a}} \ind{\wt\cap\ball{\y}{a}\ne \phi}\dd\y}\nonumber\\
&=\int_{\rthree\setminus \ball{0}{a}} \prob{\wt\cap\ball{\y}{a}\ne \phi}\dd\y.\label{eq:meanKt}
\end{align}
\end{small}
$\prob{\wt\cap\ball{\y}{a}\ne \phi}$ is the probability that the  path of the IM visits the $a$-neighborhood of the point $\y$ at least once  during time $t$. This  equals the probability that a molecule undergoing Brownian motion with initial point $\mathbf{0}$ reaches the fully-absorbing sphere of radius $a$ around the $\y$ point within time $t$ and is given by \cite{Yilmaz2014},
\begin{align}
\prob{\wt\cap\ball{\y}{a}\ne \phi}&=\frac{a}{\norm{\y}}\erfc\left(\frac{\norm{\y}-a}{\sqrt{4Dt}}\right)\label{e9}.
\end{align}
Substituting \eqref{e9} in \eqref{eq:meanKt}, then using  $ \erfc(.) $ function definition, and finally changing order of integration gives Lemma 1. The proof and the final expression is similar to  \cite[eq. 11]{Berezhkovskii1989}.
\section{Proof of Remark 1}\label{app:B}

The  time $\tc$  for which the hitting probability $ \pH{t} $ attains $k=1-1/e $ times the maximum value (which is 1) is given as
	\begin{align}
	\pH{\tc}= 1-\exp\left(-4\pi\lambda Da\tc-8a^2\lambda\sqrt{\pi D \tc}\right)=1-e^{-1}.\label{ebb1}
	\end{align}
Solving \eqref{ebb1} for $ t_c $  gives \eqref{rem1a}.
Now,
\begin{align*}
e^{-4\pi\lambda Da\tc}&\ge e^{-4\pi\lambda Da\tc-8a^2\lambda\sqrt{\pi D \tc}}=1-k.\\
\tc&\le \frac{\ln\left(1/(1-k)\right)}{4\pi\lambda Da}= \frac{1}{4\pi\lambda Da}.
\end{align*}
Similarly,
\begin{align*}
e^{-8a^2\lambda\sqrt{\pi D \tc}}&\ge e^{-4\pi\lambda Da\tc-8a^2\lambda\sqrt{\pi D \tc}}=1-k.\nonumber\\
\tc &\le \left(\frac{\ln\left(1/(1-k)\right)}{8a^2\lambda\sqrt{\pi D}}\right)^2=\frac{1}{64a^4\lambda^2{\pi D}}.
\end{align*}
Combining the two, we get \eqref{rem1}.
\section{Proof of Remark 3}
\label{app:C}
Applying the inequality $ \erf(x)\geq 1-\exp\left(-x^2\right) $ \cite{erfin} in \eqref{eq:e17} gives
\begin{align}
	\pH{\mu,\tc}&\approx\left(1-\exp\left(-2\beta\sqrt{\alpha \tc}-\alpha \tc\right)\right)\times\left(1-\frac{\mu}{\alpha}\right)+\frac{\sqrt{\pi}\beta \mu}{\alpha}\left(1-\exp\left(-2\beta\sqrt{\alpha \tc}-\alpha \tc\right)\right).
\end{align}
Note that the maximum value attainable by $\pH{\mu,t}$ is  $p_\infty= 1+\frac{\mu}{\alpha}\left(\sqrt{\pi}\beta-1\right) $. 
Now, proceeding the derivation as similar to Appendix \ref{app:B}, with $ k=p_\infty\left(1-e^{-1}\right)$ gives \eqref{rem3a} and \eqref{rem3}.
\bibliographystyle{IEEEtran}
\bibliography{letter_ref}

\end{document}